\documentclass[aps,pra,showpacs]{article}

\usepackage[utf8]{inputenc} 
\usepackage{amsmath}
\usepackage{amssymb}
\usepackage{mathrsfs}
\usepackage{bm}
\usepackage{color}
\usepackage{float}
\usepackage{amsthm}
\usepackage{appendix}
\usepackage{url}
\usepackage{fullpage}
\usepackage{authblk}

\usepackage{adjustbox}
\usepackage[font=footnotesize, skip = 1pt, justification=justified]{caption}
\usepackage[font=footnotesize, justification=justified]{subcaption}
\usepackage{enumitem}
\usepackage{graphicx}

\newcommand{\ket}[1]{| #1 \rangle}

\newtheorem{theorem}{Theorem}

\usepackage{tikz}
\usetikzlibrary{shapes,arrows}
\tikzstyle{block} = [draw, fill=white, rectangle, 
    minimum height=3em, minimum width=6em]
\tikzstyle{sum} = [draw, fill=white, circle, node distance=1cm]
\tikzstyle{input} = [coordinate]
\tikzstyle{output} = [coordinate]
\tikzstyle{pinstyle} = [pin edge={to-,thin,black}]
\usetikzlibrary{positioning}

\usepackage{tabularx}
\newcounter{protocol}
\makeatletter
\newenvironment{protocol}[1]
  {\refstepcounter{protocol} \par\addvspace{\topsep}
   \noindent
   \tabularx{\linewidth}{@{} X @{}}
    \hline
     \protected@edef\@currentlabelname{#1} 
    \textbf{Protocol \theprotocol} #1 \\
    \hline}
  { 
   \endtabularx
   \par\addvspace{\topsep}}
    \def\TX@find@end#1{%
    \def\@tempa{#1}%
    \ifx\@tempa\TX@%
       \toks@\expandafter{\the\toks@\AddBeforeEndtabularx}%
       \expandafter\TX@endtabularx
    \else\toks@\expandafter
       {\the\toks@\end{#1}}\expandafter\TX@get@body\fi}
   \makeatother

\def\AddBeforeEndtabularx{\\ \hline}

\title{A Private Quantum Bit String Commitment}

\author[1,2]{Mariana Gama}
\author[1,2]{Paulo Mateus}
\author[1,3]{Andr\'e Souto}

\affil[1]{Instituto de Telecomunica\c{c}\~{o}es, Lisbon, Portugal}
\affil[2]{Department of Mathematics,  Instituto  Superior T\'ecnico, Lisbon, Portugal}
\affil[3]{Department of Informatics, Faculty of Sciences, Lisbon University, Lisbon, Portugal}
\begin{document}
\maketitle

\begin{abstract}
We propose an entanglement-based quantum bit string commitment protocol whose composability is proven in the random oracle model. This protocol has the additional property of preserving the privacy of the committed message. Even though this property is not resilient against man-in-the-middle attacks, this threat can be circumvented by considering that the parties communicate through an authenticated channel. The protocol remains secure (but not private) if we realize the random oracles as physical unclonable functions in the so-called bad PUF model with access before the opening phase. 

\end{abstract}

\maketitle

\section{Introduction}

One of the most basic building blocks of complex cryptosystems is commitment schemes. A commitment scheme is a protocol that allows two mistrustful parties to interact in order to communicate some information that is set up a priori by the sender and that the receiver can only unveil at a later stage.  In other words, it is just as if the message was sent inside a locked box, which can only be opened after the sender hands the key over to the receiver.  The protocol is secure if the receiver cannot learn the message before the sender wishes to unveil it, and the sender cannot change the message after committing to it. Commitment schemes are used in several protocols, such as coin flipping, zero-knowledge proofs, and secure multiparty computation~\cite{Blum:1983:CFT:1008908.1008911,Brassard:1988:MDP:53813.53817,bc_to_ot_Damgrd2009ImprovingTS, review}. Since any weakness in the building blocks affects the security of the overall system, it is important to ensure that they are highly reliable. 

Unfortunately, classical bit commitment (BC) schemes cannot be simultaneously unconditionally secure against a corrupted sender and a corrupted receiver, and Canetti and Fischlin proved that universally composable BC is impossible in the plain model~\cite{canetti_commitment}. In 1996, Lo and Chau~\cite{Lo:1996xn} and independently Mayers~\cite{Mayers:1996te} proved a no-go theorem for unconditionally secure quantum BC in the standard non-relativistic quantum cryptographic framework. Since then, many protocols relying on additional assumptions have been presented. Although secure commitment schemes can be obtained through the exploitation of relativistic constraints, these types of protocols are challenging to implement.

In this paper, we propose a new private commitment protocol, i.e., a commitment where the message is never announced, nor can be derived from the messages exchanged between the parties. This property is attained through the use of entanglement. Since commitment protocols are mostly used as cryptographic primitives, it is of the utmost importance to study their security in different computational environments. As such, a strong emphasis is placed on the composability of these protocols. After characterizing the commitment functionality, the EPR pair trusted source functionality and the random oracle functionality in Section~\ref{preliminaries}, we show in Section~\ref{protocol} that these last two functionalities can be used as a resource to achieve a private commitment protocol with composable security, which is proven in Section~\ref{security}. In Section \ref{pufs}, we analyse the security of the protocol in the bad PUF attack model. Section~\ref{conclusions} features our final conclusions alongside with some directions for future work.

\section{Preliminaries}
\label{preliminaries}

 A bit commitment protocol starts with the \textit{commitment phase}, during which Alice chooses the value $m$ she wants to commit to, and generates the pair $(c,d)$. $c$ is the \emph{commitment}, which she immediately sends to Bob (who outputs a receipt message), and $d$ is the \emph{decommitment}, which she keeps to herself. In the \textit{opening phase}, Alice sends $(b,d)$ to Bob, who can either accept or reject. The protocol is said to be \textit{concealing} if Bob cannot learn Alice's committed message $m$ before the opening phase, and \textit{binding} if Alice cannot change her committed message $m$ after the commitment phase.
 
 The security of commitment protocols can be studied from a stand-alone perspective, with the requirements of concealingness and bindingness. However, since commitments are generally used as a subroutine of more complex tasks, it becomes mandatory for protocols to be secure in any computational environment. In a composable security proof, the parties running the protocol are considered as a single big party which must be indistinguishable from a simulated machine running an ideal functionality for commitment (see Figure~\ref{fig:Fcom}).

\begin{figure}[ht]
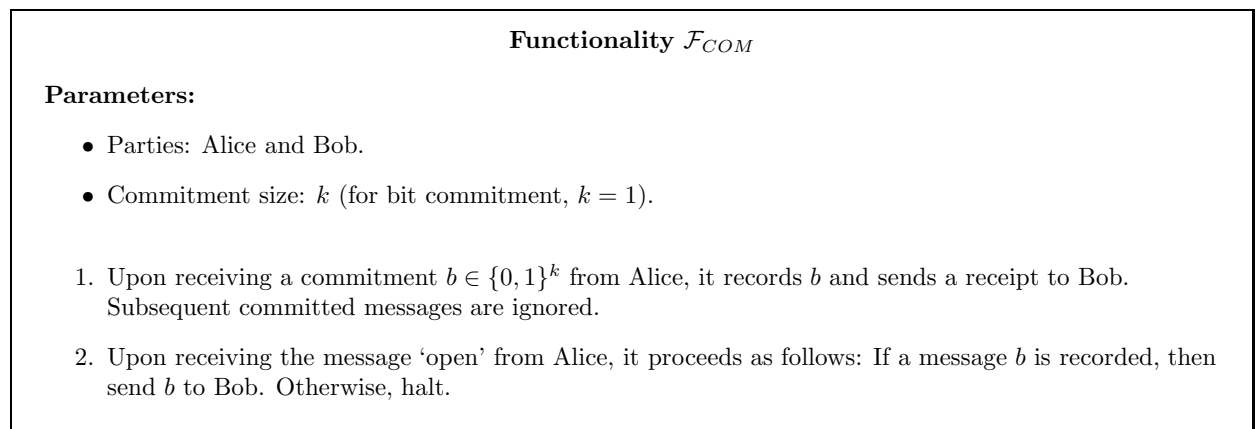

    \centering
    \adjustbox{margin=1em,width=\textwidth,set height=0.5cm,set depth=5.1cm,frame}{
    \begin{minipage}[t]{\textwidth}
    \begin{center}
            \textbf{Functionality} $\mathcal{F}_{COM}$        
    \end{center}
    \flushleft{\textbf{Parameters:}}
    \begin{itemize}
        \item Parties: Alice and Bob.
        \item Commitment size: $k$ (for bit commitment, $k=1$).
    \end{itemize}
    \vspace{0.1cm}
    \begin{enumerate}
        \item Upon receiving a commitment $b\in\{0,1\}^k$ from Alice, it records $b$ and sends a receipt to Bob. Subsequent committed messages are ignored.
        \item Upon receiving the message `open' from Alice, it proceeds as follows: If a message $b$ is recorded, then send $b$ to Bob. Otherwise, halt.
    \end{enumerate}
    
    \end{minipage}
    }
    \caption{Commitment functionality.}
    \label{fig:Fcom}
\end{figure}

In the protocol described in the next section, we assume that the parties have access to two different resources. The first one is an EPR pair trusted source modelled by the functionality in Figure~\ref{fig:F_EPR}. Note that the existence of this source is a very reasonable assumption since entanglement distribution has already been successfully implemented~\cite{zeilinger,Yin1140}. Before the beginning of the protocol, Alice and Bob can additionally sacrifice a small number of entangled pairs to estimate their correlation by using an algorithm such as the one described in Section 6.2 of~\cite{rennerphd}. Even if noisy quantum channels result in a loss of entanglement, the parties can run an entanglement distillation protocol and transform non-maximally entangled shared pairs into a smaller number of maximally entangled ones by using only local operations and classical communication (e.g.~\cite{distillation} and~\cite{distillation_zeilinger} --- the last one is significantly less effective than the first, but has the advantage of being within the reach of current technology). 

\begin{figure}[ht]
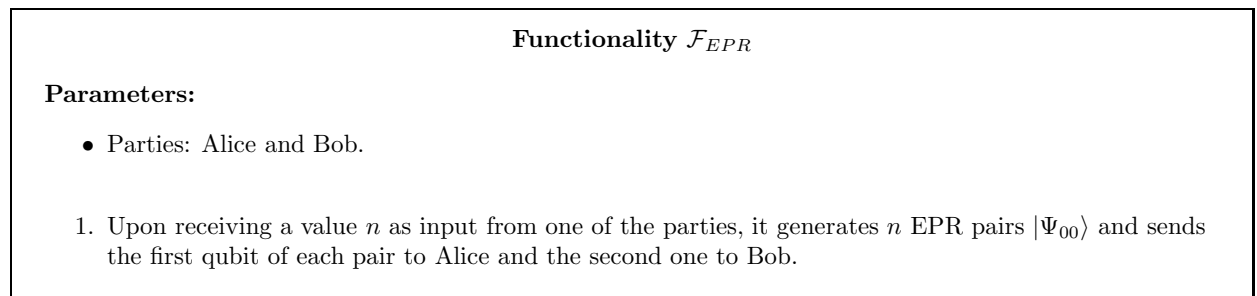

    \centering
    \adjustbox{margin=1em,width=\textwidth,set height=0.5cm,set depth=3.4cm,frame}{
    \begin{minipage}[t]{\textwidth}
    \begin{center}
            \textbf{Functionality} $\mathcal{F}_{EPR}$        
    \end{center}
    \flushleft{\textbf{Parameters:}}
    \begin{itemize}
        \item Parties: Alice and Bob.
    \end{itemize}
    \vspace{0.1cm}
    \begin{enumerate}
        \item Upon receiving a value $n$ as input from one of the parties, it generates $n$ EPR pairs $\ket{\Psi_{00}}$ and sends the first qubit of each pair to Alice and the second one to Bob.
    \end{enumerate}
    \end{minipage}
    }
    \caption{EPR pair source functionality.}
    \label{fig:F_EPR}
\end{figure}

\newpage

The second required resource, described by the functionality $\mathcal{F}_{RO}$ in Figure~\ref{fig:F_RO}, is named random oracle and behaves as an ideal cryptographic hash function, i.e., it maps each query to a fixed and uniformly random output in its range.

\begin{figure}[ht]
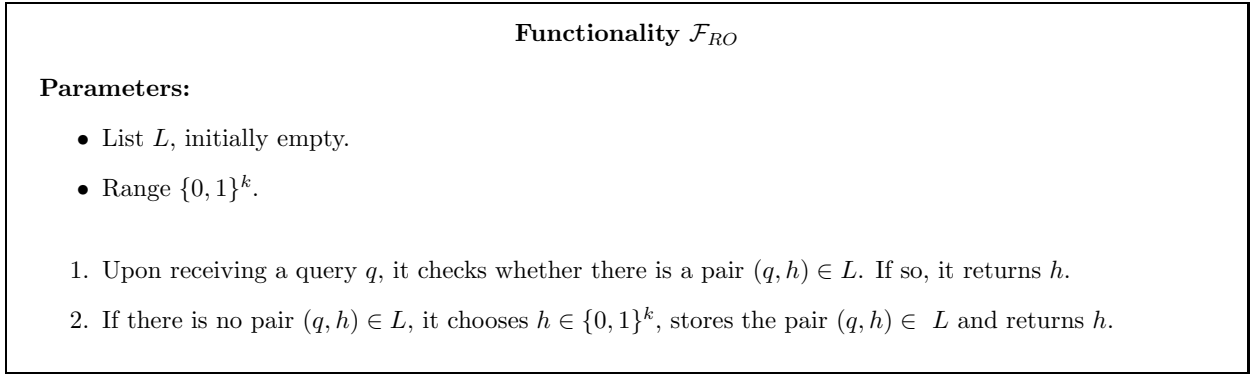

    \centering
    \adjustbox{margin=1em,width=\textwidth,set height=0.5cm,set depth=4.4cm,frame}{
    \begin{minipage}[t]{\textwidth}
    \begin{center}
            \textbf{Functionality} $\mathcal{F}_{RO}$        
    \end{center}
    \flushleft{\textbf{Parameters:}}
    \begin{itemize}
        \item List $L$, initially empty.
        \item Range $\{0,1\}^k$.
    \end{itemize}
    \vspace{0.1cm}
    \begin{enumerate}
        \item Upon receiving a query $q$, it checks whether there is a pair $(q,h)\in L$. If so, it returns $h$.
        \item If there is no pair $(q,h)\in L$, it chooses $h\in\{0,1\}^k$, stores the pair $(q,h)\in~L$ and returns $h$.
    \end{enumerate}
    \end{minipage}
    }
    \caption{Random oracle functionality.}
    \label{fig:F_RO}
\end{figure}

It is essential in our proof that a quantum computer cannot call the random oracle in superposition. Therefore a realizable random oracle implementation cannot be a cryptographic hash function such as SHA. This fact makes the random oracle quite a strong assumption; nevertheless, it can be realized using physical unclonable functions (PUFs). PUFs are physical systems with some microscale structural disorder, which is assumed to be unique to each PUF and unclonable even by the PUF manufacturer. When external stimuli (challenges) are applied to a PUF, its response will depend on the disorder of the device. Therefore, each PUF $P$ implements a unique function $f_P$ that gives responses $r=f_P (c)$ to challenges $c$. For more about PUFs we refer to \cite{puf_foundations} and \cite{puf_commitments}. PUFs have a classical interface, and cannot be run in superposition, even by an all-powerful quantum adversary.

\section{The Proposed Protocol}
\label{protocol}

\begin{protocol}{Private Quantum Bit String Commitment}
\label{theqbc2}
\textbf{Message to be shared: } $m=m_1...m_{2n}$.
\ \\[5mm]
\noindent
\textbf{Setup: }
Alice chooses a message size $2n$ and sends the value $n$ to $\mathcal{F}_{EPR}$. The functionality prepares the state $\ket{\psi} = \bigotimes_{i=1}^{n}\,\ket{\Psi_{00}}$ and sends the odd qubits to Alice and the even ones to Bob.
\ \\[5mm]
\noindent\textbf{Commitment phase:}
  \begin{enumerate}
    \item To commit to a message $m$, Alice generates an uniformly random basis string $b \in \{\{\ket{0},\ket{1}\},\{\ket{+},\ket{-}\}\}^n$, where $\ket{+}=\dfrac{\ket{0}+\ket{1}}{\sqrt{2}}$ and $\ket{-}=\dfrac{\ket{0}-\ket{1}}{\sqrt{2}}$, and measures each of her qubits $i$ in the basis $b_i$, obtaining outcomes $O \in \{0,1\}^n$. She then sends Bob the strings $c_1 = m \oplus \mathsf{H}_1(b|O)$ and $c_2 = \mathsf{H}_2(b)$, where $b|O$ is the concatenation of $b$ and $O$.
  \end{enumerate}
\textbf{Opening phase:}
  \begin{enumerate}[resume]
    \item
    Alice sends the bases $b$ to Bob.
    \item
    If $\mathsf{H}_2(b)= c_2$, Bob accepts the opening, measures each of his qubits $i$ in the basis $b_i$, obtaining outcomes $O \in \{0,1\}^n$,  and calculates $m = c_1 \oplus \mathsf{H}_1(b|O)$. Otherwise, he rejects.
  \end{enumerate}
  \vspace{0.1cm}
\end{protocol}

One of the characteristics of $\mathcal{F}_{COM}$, the functionality for commitments, is that the message is never publicly announced. In most of the existing commitment protocols, nonetheless, the opening step includes sending the message over a public channel. Here we propose a protocol (Protocol~\ref{theqbc2}) that is not only composable but also preserves the privacy of the message. We note that the privacy property is vulnerable to man-in-the-middle attacks: a third party, Eve, can pretend to be the EPR pair trusted source and send different sets of EPR pairs to Alice and Bob and then forward any received message. This can be prevented by adding an authenticated channel between Alice and Bob, as similarly done in quantum key distribution protocols. 

The protocol will use as a resource the EPR pair trusted source functionality (Figure~\ref{fig:F_EPR}) and the random oracle functionality (Figure~\ref{fig:F_RO}) presented in the previous section. It needs two instances of $\mathcal{F}_{RO}$: $\mathsf{H}_1$ with range $\{0,1\}^{2n}$ and $\mathsf{H}_2$ with range $\{0,1\}^n$. Note that, unfortunately, we cannot use the weaker version of the ROM, the global ROM~\cite{com_gro}, since the programmability of the oracle is a key point of our security proof.

\section{Security Analysis}
\label{security}

We proceed now to prove the security of Protocol~\ref{theqbc2} in the Abstract Cryptography framework~\cite{renner_ac} instantiated with quantum Turing machines~\cite{qtm}. The equivalences that need to be satisfied are depicted in Figure~\ref{fig:diagrams_ro}.


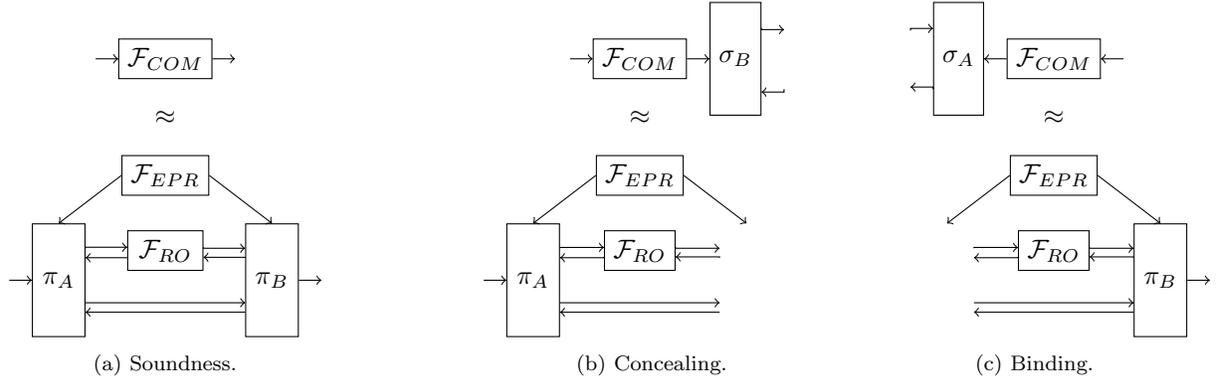
\begin{figure}[ht]
\begin{subfigure}[b]{0.3\textwidth}
\centering
\begin{tikzpicture}[every node/.style={rectangle,draw,node distance = 1cm,transform shape},scale=1]
\draw (0,0) node[minimum height=0.5cm,minimum width=1cm,draw] (1){$\mathcal{F}_{COM}$};
\node[draw=none] (blank_left1) [left=0.3cm of 1] {};
\node[draw=none] (blank_right1) [right=0.3cm of 1] {};
\draw [->] (1.east) |- (blank_right1.west) node[midway,above right,draw=none] {};
\draw [->] (1.west) (blank_left1.east)|-(1.west) node[midway,above right,draw=none] {};

\node[draw=none] (equiv) [below=0.3cm of 1] {$\approx$};
\draw  node[minimum height=0.5cm,minimum width=1cm,draw,rectangle,below = 1.01cm of 1] (epr){$\mathcal{F}_{EPR}$};
\draw  node[minimum height=0.5cm,minimum width=1cm,draw,rectangle,below = 2.01cm of 1] (2){$\mathcal{F}_{RO}$};

\draw  node[minimum height=1.5cm,minimum width=0.7cm,below left= 1.2cm and 0.8cm of equiv] (3) {$\pi_A$};
\node[draw=none] (blank_left3) [left=0.3cm of 3] {};
\draw[->] (blank_left3.east)|- (3.west)  node[midway,above right,draw=none] {};

\draw  node[minimum height=1.5cm,minimum width=0.7cm,draw,below right = 1.2 and 0.8cm of equiv] (4) {$\pi_B$};
\node[draw=none] (blank_right4) [right=0.3cm of 4] {};
\draw[->] (4.east) |- (blank_right4.west) node[midway,above right,draw=none] {};

\draw[->](epr.west) -- (3.north);
\draw[->](epr.east) -- (4.north);

\draw[->] (2.190)|-(3.40) node[midway,above left,draw=none] {};
\draw[->] (3.50)|-(2.175) node[midway,above left,draw=none] {};
\draw[->] (2.5) |- (4.130) node[midway,above right,draw=none] {};
\draw[->] (4.140)|-(2.350) node[midway,above left,draw=none] {};

\draw[->] (3.320) |- (4.220) node[midway,above left = 0.01cm and 1cm,draw=none] {};
\draw[->] (4.230) |- (3.310) node[midway,above left = 0.01cm and 1cm,draw=none] {};
\end{tikzpicture}
\caption{Soundness.}
\label{fig:honest_ro}
\end{subfigure}
\hfill
\begin{subfigure}[b]{0.3\textwidth}
\centering
\begin{tikzpicture}[every node/.style={rectangle,draw,node distance = 1cm,transform shape},scale=1]

\draw (0,0) node[minimum height=0.5cm,minimum width=1cm,draw] (1){$\mathcal{F}_{COM}$};
\node[draw=none] (blank_left1) [left=0.3cm of 1] {};
\draw node[minimum height=1.5cm,minimum width=0.5cm,right=0.3cm of 1] (blank_right1) {$\sigma_B$};
\draw [->] (1.east) |- (blank_right1.west) node[midway,above right,draw=none] {};
\draw [->] (blank_left1.east) |- (1.west) node[midway,above left,draw=none] {};

\draw node[draw=none,minimum height=1.5cm,minimum width=0.82cm,right=0.3cm of blank_right1] (pi*) {};
\draw[->] (blank_right1.50) |- (pi*.137) node[midway,above right,draw=none] {};
\draw[->] (pi*.224) |- (blank_right1.308) node[midway,above left,draw=none] {};
\node[draw=none] (blank_right_extra) [right=0.8cm of pi*] {};

\node[draw=none] (equiv) [below=0.3cm of 1] {$\approx$};
\draw  node[minimum height=0.5cm,minimum width=1cm,draw,rectangle,below = 1.01cm of 1] (epr){$\mathcal{F}_{EPR}$};
\draw  node[minimum height=0.5cm,minimum width=0.8cm,draw,rectangle,below = 2.01cm of 1] (2){$\mathcal{F}_{RO}$};

\draw  node[minimum height=1.5cm,minimum width=0.7cm,below left= 1.2cm and 0.8cm of equiv] (3) {$\pi_A$};
\node[draw=none] (blank_left3) [left=0.3cm of 3] {};
\draw[->] (blank_left3.east) |- (3.west) node[midway,above left,draw=none] {};

\draw  node[draw=none,minimum height=1.5cm,minimum width=0.7cm,below right = 1.2 and 0.8cm of equiv] (4) {};

\draw[->](epr.west) -- (3.north);
\draw[->](epr.east) -- (4.north);

\draw[->] (2.190)|-(3.40) node[midway,above left,draw=none] {};
\draw[->] (3.50)|-(2.175) node[midway,above left,draw=none] {};
\draw[->] (2.5) |- (4.130) node[midway,above right,draw=none] {};
\draw[->] (4.140)|-(2.350) node[midway,above left,draw=none] {};

\draw[->] (3.320) |- (4.220) node[midway,above left = 0.01cm and 1cm,draw=none] {};
\draw[->] (4.230) |- (3.310) node[midway,above left = 0.01cm and 1cm,draw=none] {};
\end{tikzpicture}
\caption{Concealing.}
\label{fig:bob_ro}
\end{subfigure}
\hfill \hspace{-1.5cm}
\begin{subfigure}[b]{0.3\textwidth}
\centering
\begin{tikzpicture}[every node/.style={rectangle,draw,node distance = 1cm,transform shape},scale=1]

\draw (0,0) node[minimum height=0.5cm,minimum width=1cm,draw] (1){$\mathcal{F}_{COM}$};
\draw node[minimum height=1.5cm,minimum width=0.5cm,left=0.3cm of 1] (blank_left1) {$\sigma_A$};
\node[draw=none] (blank_right1) [right=0.3cm of 1] {};
\draw [->] (blank_right1.west) |- (1.east) node[midway,above right,draw=none] {};
\draw [->] (1.west) |- (blank_left1.east) node[midway,above left,draw=none] {};

\draw node[draw=none,minimum height=1.5cm,minimum width=0.82cm,left=0.3cm of blank_left1] (pi*) {};
\draw[->] (blank_left1.230) |- (pi*.318) node[midway,above left,draw=none] {};
\draw[->] (pi*.42) |- (blank_left1.130) node[midway,above right,draw=none] {};

\node[draw=none] (equiv) [below=0.3cm of 1] {$\approx$};
\draw  node[minimum height=0.5cm,minimum width=1cm,draw,rectangle,below = 1.01cm of 1] (epr){$\mathcal{F}_{EPR}$};
\draw  node[minimum height=0.5cm,minimum width=0.8cm,draw,rectangle,below = 2.01cm of 1] (2){$\mathcal{F}_{RO}$};

\draw  node[draw=none,minimum height=1.5cm,minimum width=0.7cm,below left= 1.2cm and 0.8cm of equiv] (3) {};

\draw[->](epr.west) -- (3.north);
\draw[->](epr.east) -- (4.north);

\draw  node[minimum height=1.5cm,minimum width=0.7cm,draw,below right = 1.2 and 0.8cm of equiv] (4) {$\pi_B$};
\node[draw=none] (blank_right4) [right=0.3cm of 4] {};
\draw[->] (4.east) |- (blank_right4.west) node[midway,above right,draw=none] {};

\draw[->] (2.190)|-(3.40) node[midway,above left,draw=none] {};
\draw[->] (3.50)|-(2.175) node[midway,above left,draw=none] {};
\draw[->] (2.5) |- (4.130) node[midway,above right,draw=none] {};
\draw[->] (4.140)|-(2.350) node[midway,above left,draw=none] {};

\draw[->] (4.230) |- (3.310) node[midway,above left = 0.01cm and 1cm,draw=none] {};
\draw[->] (3.320) |- (4.220) node[midway,above left = 0.01cm and 1cm,draw=none] {};
\end{tikzpicture}
\caption{Binding.}
\label{fig:alice_ro}
\end{subfigure}
\caption{Conditions for the constructability of the resource $\mathcal{F}_{COM}$ from the resources $\mathcal{F}_{EPR}$ and $\mathcal{F}_{RO}$. Diagram (\subref{fig:honest_ro}) corresponds to the soundness property by showing the equivalence between the ideal commitment functionality $\mathcal{F}_{COM}$ and the protocol for honest parties (Alice and Bob behave according to $\pi_A$ and $\pi_B$, respectively). Diagrams (\subref{fig:bob_ro}) and (\subref{fig:alice_ro}) correspond to security against dishonest Bob and Alice, respectively. Since the algorithm they follow is unknown, $\pi_A$ and $\pi_B$ are removed from the respective real system, while the simulators $\sigma_A$ and $\sigma_B$ are respectively added to the ideal system.}
\label{fig:diagrams_ro}
\end{figure}


\begin{theorem}
Protocol~\ref{theqbc2} constructs from $\mathcal{F}_{EPR}$ and $\mathcal{F}_{RO}$ a resource that is within a negligible distance from the resource $\mathcal{F}_{COM}$ for simulators and distinguishers modelled as quantum Turing machines.
\end{theorem}

\begin{proof}
This proof will be divided into three parts, one for each of the required equivalences.

\vspace{4mm}
\noindent\textbf{Soundness}
\vspace{2mm}

Let $\ket{\psi}$ be the overall state of the system after Step 1. Note that
\begin{equation*}
    \ket{\Psi_{00}} = \frac{1}{\sqrt{2}}(\ket{00}+\ket{11}) =  \frac{1}{\sqrt{2}}(\ket{++}+\ket{--}),
\end{equation*}
so when Alice measures each of her qubits, the corresponding EPR pair will collapse to either $\ket{00}$ or $\ket{11}$ (for $b_i = \{\ket{0},\ket{1}\}$), or to either $\ket{++}$ or $\ket{--}$ (for $b_i = \{\ket{+},\ket{-}\}$). Therefore, when Bob measures each of his qubits $i$ in the basis $b'_i=b_i$ he received from Alice in the opening phase, he will get exactly the same outcome as Alice, $O'_i = O_i$, implying that $\mathsf{H}_1(b'|O')=\mathsf{H}_1(b|O)$. Bob will then retrieve the message successfully, since $c_1\oplus\mathsf{H}_1(b'|O') = m\oplus\mathsf{H}_1(b|O)\oplus\mathsf{H}_1(b'|O') = m$.

\vspace{4mm}
\noindent\textbf{Concealingness}
\vspace{2mm}

Given any behaviour of a dishonest receiver, we have to construct a simulator $\sigma_B$ that simulates $\mathsf{H}_1$, $\mathsf{H}_2$, and $\mathcal{F}_{EPR}$ and provides the receiver with a commitment that can later be opened to the message in $\mathcal{F}_{COM}$. Consider the following program for $\sigma_B$:
\begin{itemize}
    \item \emph{Simulation of $\mathsf{H}_1$:} Whenever $\sigma_B$ receives the query $b|O$ to $\mathsf{H}_1$, it answers with $h = m \oplus c_1$. In all other cases it returns a value $h$ as the ideal functionality would do and keeps $(q,h)$ on a list of queries and respective answers. 

    \item \emph{Simulation of $\mathsf{H}_2$:} Whenever $\sigma_B$ receives queries $q$ to $\mathsf{H}_2$, it returns a value $h$ as the ideal functionality would do and keeps $(q,h)$ on a list of queries and respective answers. 
    
    \item \emph{Simulation of $\mathcal{F}_{EPR}$:} During the setup phase, $\sigma_B$ generates the state $\ket{\psi}=\bigotimes_{i=1}^{n}\,\ket{\Psi_{00}}$, sends the even qubits to the corrupted receiver and keeps the odd ones to itself. 
    \item During the commitment phase, upon receiving the receipt from $\mathcal{F}_{COM}$, $\sigma_B$ chooses two uniformly random strings, $c_1\in\{0,1\}^{2n}$ and $b \in \{\{\ket{0},\ket{1}\},\{\ket{+},\ket{-}\}\}^n$, and measures each of its qubits $i$ in the basis $b_i$, obtaining outcomes $O\in \{0,1\}^n$. It then sends $c_1$ and $c_2=\mathsf{H}_2(b)$ to the corrupted receiver.
    
    \item During the opening phase, upon receiving the message $m$ from $\mathcal{F}_{COM}$, $\sigma_B$ sends the bases $b$ to the corrupted receiver.
\end{itemize}

The behaviour of $\sigma_B$ is the same regardless of the message that was sent to $\mathcal{F}_{COM}$, and hence there is no algorithm for the dishonest receiver allowing him to guess the committed message with probability greater than $1/2^{2n}$.

\vspace{4mm}
\noindent\textbf{Bindingness}
\vspace{2mm}

Given any behaviour of a dishonest sender, we have to construct a simulator $\sigma_A$ that simulates $\mathsf{H}_1$, $\mathsf{H}_2$, and $\mathcal{F}_{EPR}$ and retrieves the message $m$ from the sender's commitment values and sends it to $\mathcal{F}_{COM}$. It must also be able to detect when the sender is cheating and, whenever that happens, not send the opening message to $\mathcal{F}_{COM}$. Consider the following program for $\sigma_A$:
\begin{itemize}
    \item \emph{Simulation of $\mathsf{H}_1$ and $\mathsf{H}_2$:} Whenever $\sigma_A$ receives queries $q$ to $\mathsf{H}_1$ or $\mathsf{H}_2$, it returns a value $h$ as the ideal functionality would do and keeps $(q,h)$ on a list of queries and respective answers.
    
    \item \emph{Simulation of $\mathcal{F}_{EPR}$:} During the setup phase, $\sigma_A$ generates the state $\ket{\psi}=\bigotimes_{i=1}^{n}\,\ket{\Psi_{00}}$, sends the odd qubits to the corrupted sender and keeps the even ones to itself.
    
    \item During the commitment phase, upon receiving the commitment strings $c_1$ and $c_2$ from the corrupted sender, $\sigma_A$ sends $m = c_1 \oplus \mathsf{H}_1 (b|O)$ to $\mathcal{F}_{COM}$.
    
    \item During the opening phase, upon receiving the basis string $b'$ from the corrupted sender, $\sigma_A$ sends the message `open' to $\mathcal{F}_{COM}$ if $b'=b$. Otherwise, it does not open the commitment.
\end{itemize}

The real world receiver outputs error whenever the string $b'$ sent by the sender is such that $\mathsf{H}_2(b') \neq \mathsf{H}_2 (b)$. From the soundness property, we know that when $b'=b$ the receiver correctly retrieves the message. We are interested in the situation where $b' \neq b$ (in which case the commitment will not be opened in the ideal world) and $\mathsf{H}_2(b') = \mathsf{H}_2 (b)$. Since $\mathcal{F}_{RO}$ is collision-resistant, this can only happen with negligible probability.
\end{proof}

The addition of an authenticated communication channel makes this protocol a private and composable commitment protocol, which is yet to be achieved by classical cryptography based on the same assumptions.

\section{Analysis in the Realistic Bad PUF Model}
\label{pufs}



In order to study the security of PUF applications in a realistic scenario, two attack models are described in \cite{pufs_ieee}: the PUF re-use model and the bad PUF model. In the PUF re-use model, we assume some PUFs are used more than once throughout the protocol and the adversary has access to the PUFs more than once. In the bad PUF model, the fact that PUFs are real physical objects is exploited, and we consider both the simulatable bad PUFs, which possess a simulation algorithm that can be used by the manufacturer to compute responses to challenges and the challenge-logging bad PUFs, which allow the manufacturer to access a memory module in the device and read all the challenges applied to it (this malicious feature could also be added by an adversary after the construction of the PUF). The notion of strong PUFs is also described. Strong PUF is a type of PUF with a public interface (i.e., anyone holding it can apply challenges and read the responses), a large number of possible challenges and behaviour so complex that it cannot be modelled to predict responses to challenges. In our brief analysis, we consider that in the proposed protocol (Protocol \ref{theqbc2}) the ROM is replaced by strong PUFs.

Note that Protocol \ref{theqbc2} is secure in the bad PUF model if we also consider that Alice sends the message $m$ to Bob in the opening phase, thus giving up the privacy property.  The security holds independently of whether the malicious party has access to the PUFs before the opening phase or not. This follows from the fact that considering the PUFs are manufactured by Alice and she can find collisions in $\mathsf{H}_1$ and $\mathsf{H}_2$, she will still not know what message to open to in order to match the one calculated by Bob since the outcomes of his measurements of qubits in incorrect bases will be uniformly random.




\section{Conclusions}
\label{conclusions}

With this work, we achieved a commitment protocol that is not only composable but also private, since the message is never publicly announced. Man-in-the-middle attacks can be prevented by adding an authenticated channel. We suggest the use of physical unclonable functions to model random oracles, and note that the protocol remains secure (although not private) if we consider the bad PUF attack model with access before the opening phase, which has been proven impossible for classical bit commitment without other assumptions. 

Additionally, it is also of interest to further study how to obtain composability in commitment schemes while using the minimum possible assumptions (for more on this topic see~\cite{minimal_assumptions}), and which of these assumptions are needed to achieve privacy.

\section*{Acknowledgments}
The authors acknowledge the support of SQIG (Security and Quantum Information Group), the Instituto de Telecomunicações (IT) Research Unit, Ref. UIDB/EEA/50008/2020, funded by Fundação para a Ciência e Tecnologia e Ministério Ciêncica, Tecnologia e Ensino Superior (FCT/MCTES), and the FCT projects Confident PTDC/EEI-CTP/4503/2014, QuantumMining POCI-01-0145-FEDER-031826, and Predict PTDC/CCI-CIF/29877/2017, supported by the European Regional Development Fund (FEDER), through the Competitiveness and Internationalization Operational Programme (COMPETE 2020), and by the Regional Operational Program of Lisbon. A.S. acknowledges funds granted to Laboratório de Sistemas Informáticos de Grande Escala (LASIGE) Research Unit, Ref. UIDB/00408/2020. M.G. also acknowledges the support of the Calouste Gulbenkian Foundation through the New Talents in Quantum Technologies Programme.

\newpage

\bibliographystyle{alpha}
\bibliography{references}

\newcommand{\etalchar}[1]{$^{#1}$}
\begin{thebibliography}{LYM{\etalchar{+}}19}

\bibitem[ALP{\etalchar{+}}14]{review}
{\'A}lvaro~J. {Almeida}, Ricardo {Loura}, Nikola {Paunkovi{\'c}}, Nuno~A.
  {Silva}, Nelson~J. {Muga}, Paulo {Mateus}, Paulo~S. {Andr{\'e}}, and
  Armando~N. {Pinto}.
\newblock {\em {A brief review on quantum bit commitment}}, volume 9286 of {\em
  Society of Photo-Optical Instrumentation Engineers (SPIE) Conference Series},
  page 92861C.
\newblock 2014.

\bibitem[BBP{\etalchar{+}}96]{distillation}
Charles~H. Bennett, Gilles Brassard, Sandu Popescu, Benjamin Schumacher,
  John~A. Smolin, and William~K. Wootters.
\newblock Purification of noisy entanglement and faithful teleportation via
  noisy channels.
\newblock {\em Phys. Rev. Lett.}, 76:722--725, Jan 1996.

\bibitem[BCC88]{Brassard:1988:MDP:53813.53817}
Gilles Brassard, David Chaum, and Claude Cr{\'e}peau.
\newblock Minimum disclosure proofs of knowledge.
\newblock {\em J. Comput. Syst. Sci.}, 37(2):156--189, Oct 1988.

\bibitem[Blu83]{Blum:1983:CFT:1008908.1008911}
Manuel Blum.
\newblock Coin flipping by telephone a protocol for solving impossible
  problems.
\newblock {\em SIGACT News}, 15(1):23--27, Jan 1983.

\bibitem[CF01]{canetti_commitment}
Ran Canetti and Marc Fischlin.
\newblock Universally composable commitments.
\newblock In Joe Kilian, editor, {\em Advances in Cryptology --- CRYPTO 2001},
  pages 19--40, Berlin, Heidelberg, 2001. Springer Berlin Heidelberg.

\bibitem[CJS14]{com_gro}
Ran Canetti, Abhishek Jain, and Alessandra Scafuro.
\newblock Practical uc security with a global random oracle.
\newblock 11 2014.

\bibitem[DFL{\etalchar{+}}09]{bc_to_ot_Damgrd2009ImprovingTS}
Ivan Damg{\aa}rd, Serge Fehr, Carolin Lunemann, Louis Salvail, and Christian
  Schaffner.
\newblock Improving the security of quantum protocols via commit-and-open.
\newblock In {\em CRYPTO}, 2009.

\bibitem[LC97]{Lo:1996xn}
Hoi-Kwong Lo and H.~F. Chau.
\newblock {Is quantum bit commitment really possible?}
\newblock {\em Phys. Rev. Lett.}, 78:3410, 1997.

\bibitem[LYM{\etalchar{+}}19]{minimal_assumptions}
M.~{Lemus}, P.~{Yadav}, P.~{Mateus}, N.~{Paunković}, and A.~{Souto}.
\newblock On minimal assumptions to obtain a universally composable quantum bit
  commitment.
\newblock In {\em 2019 21st International Conference on Transparent Optical
  Networks (ICTON)}, pages 1--4, July 2019.

\bibitem[May96]{Mayers:1996te}
Dominic Mayers.
\newblock {Unconditionally secure quantum bit commitment is impossible}.
\newblock 1996.
\newblock [Phys. Rev. Lett.78,3414(1997)].

\bibitem[MR11]{renner_ac}
Ueli Maurer and Renato Renner.
\newblock Abstract cryptography.
\newblock In {\em IN INNOVATIONS IN COMPUTER SCIENCE}. Tsinghua University
  Press, 2011.

\bibitem[MSS15]{qtm}
P.~Mateus, A.~Sernadas, and A.~Souto.
\newblock {Universality of quantum Turing machines with deterministic control}.
\newblock {\em Journal of Logic and Computation}, 27(1):1--19, 02 2015.

\bibitem[PSBZ01]{distillation_zeilinger}
Jian-Wei Pan, Christoph Simon, {\v C}aslav Brukner, and Anton Zeilinger.
\newblock Entanglement purification for quantum communication.
\newblock {\em Nature}, 410:1067--1070, 2001.

\bibitem[Ren05]{rennerphd}
Renato Renner.
\newblock {\em Security of quantum key distribution}.
\newblock PhD thesis, ETH Zurich, 2005.

\bibitem[RSS09]{puf_foundations}
Ulrich Rührmair, Jan Sölter, and Frank Sehnke.
\newblock On the foundations of physical unclonable functions.
\newblock {\em IACR Cryptology ePrint Archive}, 2009:277, 01 2009.

\bibitem[Rv13]{pufs_ieee}
U.~{Rührmair} and M.~{van Dijk}.
\newblock Pufs in security protocols: Attack models and security evaluations.
\newblock In {\em 2013 IEEE Symposium on Security and Privacy}, pages 286--300,
  May 2013.

\bibitem[vDR12]{puf_commitments}
Marten van Dijk and Ulrich R{\"u}hrmair.
\newblock Physical unclonable functions in cryptographic protocols: Security
  proofs and impossibility results.
\newblock {\em IACR Cryptology ePrint Archive}, 2012:228, 2012.

\bibitem[WJS{\etalchar{+}}19]{zeilinger}
S{\"o}ren Wengerowsky, Siddarth~Koduru Joshi, Fabian Steinlechner, Julien~R.
  Zichi, Sergiy~M. Dobrovolskiy, Ren{\'e} van~der Molen, Johannes W.~N. Los,
  Val Zwiller, Marijn A.~M. Versteegh, Alberto Mura, Davide Calonico, Massimo
  Inguscio, Hannes H{\"u}bel, Liu Bo, Thomas Scheidl, Anton Zeilinger,
  Andr{\'e} Xuereb, and Rupert Ursin.
\newblock Entanglement distribution over a 96-km-long submarine optical fiber.
\newblock {\em Proceedings of the National Academy of Sciences},
  116(14):6684--6688, 2019.

\bibitem[YCL{\etalchar{+}}17]{Yin1140}
Juan Yin, Yuan Cao, Yu-Huai Li, Sheng-Kai Liao, Liang Zhang, Ji-Gang Ren,
  Wen-Qi Cai, Wei-Yue Liu, Bo~Li, Hui Dai, Guang-Bing Li, Qi-Ming Lu, Yun-Hong
  Gong, Yu~Xu, Shuang-Lin Li, Feng-Zhi Li, Ya-Yun Yin, Zi-Qing Jiang, Ming Li,
  Jian-Jun Jia, Ge~Ren, Dong He, Yi-Lin Zhou, Xiao-Xiang Zhang, Na~Wang, Xiang
  Chang, Zhen-Cai Zhu, Nai-Le Liu, Yu-Ao Chen, Chao-Yang Lu, Rong Shu,
  Cheng-Zhi Peng, Jian-Yu Wang, and Jian-Wei Pan.
\newblock Satellite-based entanglement distribution over 1200 kilometers.
\newblock {\em Science}, 356(6343):1140--1144, 2017.

\end{thebibliography}

\end{document}